\documentclass[reprint]{revtex4-1}
\pdfoutput=1
\usepackage{amsmath}
\usepackage{graphicx}
\usepackage{placeins}
\usepackage{amsthm}
\usepackage{nicefrac}
\usepackage{bbm}
\usepackage{wrapfig}
\usepackage{amssymb}
\usepackage{pgf,tikz}
\usepackage{mathrsfs}
\usetikzlibrary{arrows}
\usepackage{mathrsfs}
\usepackage{braids}
\usepackage{fullpage}
\usepackage[title]{appendix}
\usepackage{epigraph}
\usepackage{listings}
\usepackage{multirow}
\newtheorem{thm}{Theorem}
\newtheorem{cor}{Corollary}
\newtheorem{lem}{Lemma}
\theoremstyle{definition}
\newtheorem{defn}{Definition}

\newcommand{\ket}[1]{{\left\vert{#1}\right\rangle}}

\definecolor{xdxdff}{rgb}{0.49019607843137253,0.49019607843137253,1.}
\definecolor{uuuuuu}{rgb}{0.26666666666666666,0.26666666666666666,0.26666666666666666}
\definecolor{qqqqff}{rgb}{0.,0.,1.}
\definecolor{ffqqqq}{rgb}{1.,0.,0.}
\begin{document}
\title{On the computational complexity of detecting possibilistic locality}
\author{Andrew W. Simmons}\affiliation{Department of Physics, Imperial College London, SW7 2AZ.}
\begin{abstract}
The proofs of quantum nonlocality due to GHZ and Hardy are quantitatively different from that of Bell insofar as they rely only on a consideration of whether events are possible or impossible, rather than relying on specific experimental probabilities. Here, we consider the computational task of determining whether or not a given table of possibilities constitutes a departure from possibilistic local realism. By considering the case in which one party has access to measurements with two outcomes and the other three, it is possible to see at exactly which point this task becomes computationally difficult.
\end{abstract}
\maketitle

\section{Introduction}
Since the inception of quantum mechanics, the apparent nonlocality of the theory has been a thorny issue. Einstein, together with Podolsky and Rosen \cite{EPR}, argued that this property, in apparent contravention of the information-propagation bound of the speed of light from special relativity, was evidence of the incompleteness of the quantum-mechanical description of nature. The related concept of nonlocal steering was also the subject of Einstein's scrutiny in a letter to Edwin Schr\"{o}dinger \cite{Eins1935}, in which it is noted that a local description of steering necessitates nature to have what we would now call the $\psi$-epistemic property: that the quantum state is not instantiated as a part of nature's ontology.

In this paper, we will consider ``Bell experiments'' of the classic form: we have two separate quantum systems, in a possibly entangled state, and a set of observables that we can measure on each. Formally, this is a prepare-and-measure scenario for which the commutativity graph of the available observables is bipartite; these notions are equivalent since we will consider systems of finite quantum dimension. Such an experiment will have probabilities for its outcomes predicted by quantum theory; we are then interested in whether or not these predictions could be realised via a local (that is, factorisable) hidden variable model. That this is not always possible was first noted by Bell \cite{Bell1966}.

In 1991, Pitowsky \cite{Pito1991} demonstrated that determining whether or not a given table of these operational predictions had a local hidden variable explanation was \textbf{NP}-Hard. In this paper, we consider a more extreme, possibilistic, manifestation of nonlocality; one in which it is not only true that we cannot account for the specific experimental probabilities with a local hidden variable model, but that we cannot even account for the possibilities. In doing so, we will identify the minimal operational requirements for this problem to remain \textbf{NP}-Hard, filling in the final gap in a categorisation started by Mansfield and Fritz \cite{Mans2011,Mans2016} and continued by Abramsky, Gottlob and Kolaitis \cite{Abra2013}.

\section{The Locality Decision Problem}

Given a table of probabilities for a Bell experiment, a question of experimental and foundational relevance is whether or not that table provides an example of quantum nonlocality; that is to say that it cannot be reproduced by a local hidden variable model. A table of probabilities can be transformed into a table of possibilities simply by denoting each nonzero element with a 1 to mark it as possible. On the top of figure \ref{CHSH}, we see the standard table of probabilities given; below is the corresponding table of possibilities. We will follow the notation of Mansfield and Fritz for these tables, in which each row and column denotes a measurement outcome, and the intersection of a row and column is the probability of that outcome within the contexts, which are separated by lines. For example, if Alice chooses measurement setting $A_2$, and Bob chooses measurement setting $B_1$, then we see from the table that the probability of seeing the outcome $(a_2=0, b_2=1)$ is $\nicefrac18$.

\begin{figure}
\begin{equation*}
\begin{array}{c| c c | c c |}
&b_1=0&b_1=1&b_2=0&b_2=1\\ \hline
a_1=0&\nicefrac{1}{2} &0  & \nicefrac{3}{8} & \nicefrac{1}{8}\\
a_1=1&0&\nicefrac{1}{2}  &\nicefrac{1}{8} &\nicefrac{3}{8} \\ \hline
a_2=0&\nicefrac{3}{8}&\nicefrac{1}{8} &\nicefrac{1}{8} &\nicefrac{3}{8}  \\
a_2=1&\nicefrac{1}{8}&\nicefrac{3}{8} &\nicefrac{3}{8} & \nicefrac{1}{8} \\ \hline
\end{array} \end{equation*}\begin{equation*}
\begin{array}{c| c c | c c |}
&b_1&b_1'&b_2&b'_2\\ \hline
a_1&1 &0  &1 & 1\\
a_1'&0&1  &1 &1\\ \hline
a_2&1&1&1 &1 \\
a_2'&1&1 &1 &1 \\ \hline
\end{array} 
\end{equation*}
\caption{Probability and Possibility tables for a CHSH experiment. Alice has two available measurements: measurement $A_1$ with outcome $a_1$ and $a_1'$; and measurement $A_2$ with outcomes $a_2$ and $a'_2$. Bob has measurements denoted similarly. Typically, however, we will not label outcomes as the properties we are consider are invariant under relabelling.}
\label{CHSH}
\end{figure}

\begin{defn}[$(j,k)$-\textsc{PossLoc}]
The $k$-outcome Possibilistic locality decision problem, $(j,k)$-\textsc{PossLoc} is as follows: the problem instance is an element of $\{0,1\}^{j+k+n+m}$, where $j$ and $k$ are the maximum number of measurement outcomes for each measurement on the two subsystems being measured, and $n$ and $m$ are the number of measurements available at each location. This should be thought of as a data table where a 1 indicates a possible outcome and a 0 indicates an impossible outcome. In general, we will take $n=m$ to maintain a single scaling factor; in fact this change is without loss of generality. In the case where $j=k$, we will denote this problem merely $k$-\textsc{PossLoc}. While this is not essential, we will also require that the data tables obey a possibilistic no-signalling principle; this merely streamlines the wording of the theorems. 
\end{defn}
\begin{thm} \cite{Mans2011}
2-\textsc{PossLoc} is in \textbf{P}.
\end{thm}
\begin{proof}
This proof proceeds by showing that the only kind of possibilistic nonlocality that can occur in the two-outcome case is the one originally identified by Hardy \cite{Hard1993}; that is to say, a feature of the following kind, in which a blank space can be substituted either for a 0 or for a 1:
\begin{equation*}
\begin{array}{c| c c | c c |}
&&&&\\ \hline
&1 & & 0 & \\
&& & & \\ \hline
&0& & &  \\
&& & & 0 \\ \hline
\end{array}
\end{equation*}
We can quickly see that the highlighted one is impossible to extend to a ``deterministic grid'', \emph{i.e.} no single deterministic hidden variable can account for that possibility without also introducing an event that is impossible in the diagram. Since the appearance of this structure is equivalent to possibilistic nonlocality in this case, an algorithm to determine whether or not a given data table is possibilistically local or nonlocal reduces to checking each possible set of four contexts of this type, of which there are $n^2(n-1)^2/4$. Since checking for the appearance of such structures is possible in constant time, our algorithm runs in $O(n^4)$.
\end{proof}
However, this argument is specific to the case in which we have measurements with at most two outcomes on each side. As an illustration that there are quantumly-accessible nonlocal data tables that are not reducible to a fine-graining of a Hardy paradox, we will now demonstrate a novel possibilistic nonlocality scenario in the (2,3) case. This will be a generalisation of Hardy's proof of Bell's theorem.

\begin{figure}
\begin{center}
\begin{tikzpicture}[line cap=round,line join=round,>=triangle 45,x=0.8cm,y=0.8cm]
\clip(-5.,-4.) rectangle (5.,5.);
\draw(0.,0.) circle (3.2cm);
\draw (-2.5074154767005594,3.116547388890871)-- (2.5262219891238353,3.101322695507066);
\draw (-0.6234859067055393,3.9511093789136704)-- (3.70768134486122,-1.501032659527465);
\draw (0.6202539684631545,3.951618024886225)-- (-3.6765801200722317,-1.575677194316671);
\draw [dash pattern=on 2pt off 2pt] (-2.5074154767005594,3.116547388890871)-- (-3.6765801200722317,-1.575677194316671);
\draw [dash pattern=on 2pt off 2pt] (-3.6765801200722317,-1.575677194316671)-- (3.70768134486122,-1.501032659527465);
\draw [dash pattern=on 2pt off 2pt] (-3.6765801200722317,-1.575677194316671)-- (2.5262219891238353,3.101322695507066);
\draw [dash pattern=on 2pt off 2pt] (-2.5074154767005594,3.116547388890871)-- (3.70768134486122,-1.501032659527465);
\draw [dash pattern=on 2pt off 2pt] (3.70768134486122,-1.501032659527465)-- (2.5262219891238353,3.101322695507066);
\begin{scriptsize}
\draw [fill=uuuuuu] (0.,0.) circle (1.5pt);
\draw [fill=xdxdff] (-2.5074154767005594,3.116547388890871) circle (1.5pt);
\draw[color=xdxdff] (-2.8168692003229885,3.7112493895076946) node {$x$};
\draw [fill=xdxdff] (2.5262219891238353,3.101322695507066) circle (1.5pt);
\draw[color=xdxdff] (2.7461074833053605,3.5513937376792932) node {$y$};
\draw [fill=xdxdff] (-0.6234859067055393,3.9511093789136704) circle (1.5pt);
\draw[color=xdxdff] (-0.802687987285138,4.47855651828402) node {$a$};
\draw [fill=xdxdff] (3.70768134486122,-1.501032659527465) circle (1.5pt);
\draw[color=xdxdff] (4.024952697932568,-1.2762469475384213) node {$A$};
\draw [fill=xdxdff] (0.6202539684631545,3.951618024886225) circle (1.5pt);
\draw[color=xdxdff] (0.8598107917302307,4.4146142575526595) node {$b$};
\draw [fill=xdxdff] (-3.6765801200722317,-1.575677194316671) circle (1.5pt);
\draw[color=xdxdff] (-4.127685545315876,-1.244275817172741) node {$B$};
\end{scriptsize}
\end{tikzpicture}
\end{center}
\caption{Points on the Bloch sphere to which Alice's ensemble can be steered, revealing the generalised Hardy paradox.}
\label{steer}
\end{figure}
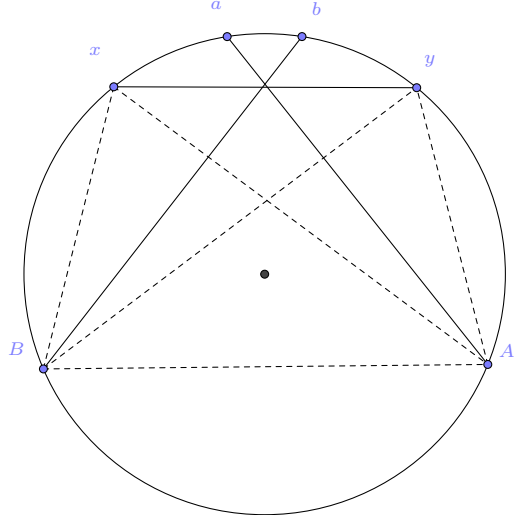

Constructing the possibility table for this steering scenario reveals the following structure in a subset of the table:
\begin{equation*}
\begin{array}{ c  | c c | c c c| c c c|}\
  & b&b^\perp& A^\perp &B^\perp&x^\perp& A^\perp &B^\perp&y^\perp \\ \hline
  a &1&1&1&1&1&1&1&1\\ 
  A &1&1&0&1&1&0&1&1\\ \hline
  b &1&0&1&1&1&1&1&1\\ 
 B  &1&1&1&0&1&1&0&1\\ \hline
 x &1&1&1&1&0&1&1&1\\ 
 y &1&1&1&1&1&1&1&0\\ \hline
\end{array}
\end{equation*}
We now present this paradox minimally, demonstrating that it is neither a Hardy paradox, nor a fine-graining of a Hardy paradox. Therefore, the Hardy paradox is not universal even for quantumly-accessible possibilistic nonlocality as long as one party has access to three-outcome measurements. 

\begin{equation*}
\begin{array}{ c | c c | c c c| c c c|}
&&&&&&&& \\ \hline
&1&&0&&&0&&\\
   &&&&&&&&\\ \hline
   &0&&&&&&&\\ 
  &&&&0&&&0&\\ \hline
  &&&&&0&&&\\ 
  &&&&&&&&0\\ \hline
\end{array}
\end{equation*}
We see that in some sense this can be thought of as a generalisation of Hardy's paradox; the third measurement row acts to convert the two right-hand measurement columns into a single effective column identical to that of the standard Hardy paradox; we see that at least one of the two right-hand columns must have its completion of the highlighted 1 in a subcolumn other than the rightmost subcolumn, causing a Hardy paradox. We note also that this family of no-signalling distributions has a supremum value of the \emph{paradoxical probability} (the maximum probability of a measurement witnessing a Hardy paradox) of $\nicefrac{1}{2}$. This is a large violation compared to other known situations with a fixed number of outcomes, a fixed quantum dimension, and a fixed number of measurements: the paradoxical probability for Hardy's own construction is roughly 0.09. This probability can be met by scenarios with two parties with access to $k$ different 2-outcome measurements, or by a generalised no-signalling theory \cite{Mans2016}. This is because the overlap between $\ket{A}$ and $\ket{b^\perp}$ can be made arbitrarily large, and the reduced state of the system can be made arbitrarily close to the completely mixed state.

\section{\textsc{(2,3)-PossLoc} is \textbf{NP}-complete}

It is easy to see that $(j,k)$-\textsc{PossLoc} is in \textbf{NP}; to show that a data table is possibilistically local, we can specify a possibilistically local hidden variable model that accounts for each of the occurrences of the possible events (1 entries) within the problem instance. It is clear to see that such a witness is only polynomially-sized, and can be checked in polynomial time. The rest of this section, then, will be a proof of \textbf{NP}-hardness. We note that since Mansfield and Fritz \cite{Mans2011} showed that \textsc{2-PossLoc} was in \textbf{P} and Abramsky \emph{et al} showed that \text{3-PossLoc}  and the three-party generalisation of \textsc{2-PossLoc} were \textbf{NP}-complete, this leaves $(j,k)$-\textsc{PossLoc} as the only remaining case to have its complexity analysed.

\begin{defn}[$r$-\textsc{Robust} decision problems]	
Following Abramsky \emph{et al}, we shall define, for a decision problem \textsc{P}, the problem $r$-\textsc{Robust P} to be the decision problem that, given an instance of \textsc{P}, asks whether or not every assignment of $r$ variables in the problem can be extended to a satisfying assignment of the problem.
\end{defn}

\begin{thm}
The \textsc{(2,3)-PossLoc} decision problem is \textbf{NP}-complete.
\end{thm}
\begin{proof}
\begin{defn}
\textsc{0(1)-valid 3-SAT} is the set of \textsc{3-SAT} decision problems that are satisfied by assigning a value of 0 (1) to all variables.
\end{defn}

\begin{lem}\label{rob3sat}
\textsc{2-Robust 0-valid 1-valid 3-SAT} is \textbf{NP}-complete.
\end{lem}
\begin{proof}[Proof of Lemma]
Consider a general \textsc{3-SAT} instance $C=\bigwedge_i c_i$, where $c_i$ are clauses consisting of the disjunction of three literals or negated literals. We will introduce two new variables $x$ and $y$, and apply the mapping $f$ to each clause $c_i$ defined by:
\begin{equation}
f(l_1\cup l_2 \vee  \circ  l_3) = (l_1\vee l_2 \vee \circ l_3 \vee \neg x),
\end{equation}
\begin{equation}
f(\neg l_1\vee\neg l_2 \vee \circ  l_3) = (\neg l_1\cup\neg l_2 \vee \circ  l_3 \vee  y),
\end{equation}
in which the $\circ$ symbol is being used to denote the presence or absence of a $\neg$ symbol. We note that now, every clause of $C'=\bigwedge_i f(c_i)$ contains at least one positive literal and at least one negative literal, and so $C'$ is both 0-valid and 1-valid.

We note that it is possible now to convert this 4-SAT instance back to a \textsc{3-SAT} instance by the following procedure, mapping each clause to a pair of equisatisfiable clauses (which we will refer to as an effective clause):
\begin{equation}
(\circ l_1\cup\circ l_2 \vee  \circ  l_3 \vee x) \rightarrow (\circ l_1\cup\circ l_2 \vee z)\wedge(\neg z \vee \circ  l_3 \vee x)
\end{equation}

 We note also that this adjustment maintains 0-validity and 1-validity, since we can choose an order on the $\{\circ l_i\}$ such that each of the two clauses in the effective clause contain a positive and a negative literal. We note that since each effective clause contains either $\neg x$ or $y$, setting $x$ to 0 and $y$ to 1 leads to a satisfying assignment for any assignment choices of the other variables. Now, to test whether or not the instance is $2$-robust, we check that there are no choices for our two variables to fix that cause there to be no satisfying assignment. We will consider the possibilities on a case-by-case basis:
\begin{enumerate}
\item If neither $x$ is fixed to 1 or $y$ is fixed to 0, then we set $x$ to 0 and $y$ to 1. As mentioned above, this yields a satisfying assignment. If one or more $z$ has been assigned a value, we may need to assign some $l_i$ to make the affected clauses valid. Since this affects at most two clauses it can be seen that it is always possible to satisfy such a modified pair of clauses.

\item If exactly one of $x$ is fixed to 1 or $y$ is fixed to 0, then, removing all instances of the $\neg x$ or $y$ literals respectively still leaves the instance either 0-valid or 1-valid, and we can utilise this assignment.

\item If we fix both $x$ to 1 and $y$ to 0, then removing the $\neg x$ and $y$ literals from $C'$ transform it back into $C$, and therefore this restriction has a satisfying assignment if and only if $C$ did originally.
\end{enumerate}
Therefore, since \textsc{3-SAT} is \textbf{NP}-hard, so is \textsc{2-Robust 0-valid 1-valid 3-SAT}. It is in \textbf{NP} since there are only $O(n^2)$ different pairs of variables to set, and so the witness size to demonstrate a set of assignments that display robustness is sitll only polynomially sized.  Hence,  \textsc{2-Robust 0-valid 1-valid 3-SAT} is \textbf{NP}-complete.

\end{proof}
\begin{cor}
\textsc{2-Robust 3-SAT} is \textbf{NP}-complete.
\end{cor}
We note now that we will in fact only need the fact that \textsc{2-Robust 3-SAT} is \textbf{NP}-complete for the rest of this proof; the motivation for having proven that the \textsc{0-valid 1-valid} is also \textbf{NP}-complete will become apparent when considering which of these possibility tables are realisable within quantum mechanics.

\begin{lem}\label{nplem}
There is a polynomial-time embedding of \textsc{2-Robust 3-SAT} into \textsc{(2,3)-PossLoc}.
\end{lem}
\begin{proof}[Proof of Lemma]
The reduction algorithm is as follows:
\begin{itemize}
\item For each variable in the \textsc{2-Robust 3-SAT} instance, we add a measurement to the party with two-outcome measurements available to them, to which we will assign the measurement rows. We pick any ordering of the variables to do this.

\item For each clause in the \textsc{2-Robust 3-SAT} instance, we add a measurement to the party with three-outcome measurements available to them; these are our measurement columns. This is a departure from the strategy of Abramsky \emph{et al}, whose constructions have a direct symmetry between the rows and columns.

\item For each intersection of a variable row and clause column, such that the variable is represented positively in the clause, we have the measurement possibilities given by:
\begin{equation}
\begin{array}{  | c c  c|}
\hline
 0 & 1 & 1 \\
 1 & 1& 1 \\ \hline
 \end{array}\, ,
 \begin{array}{  | c c  c|}
\hline
 1 &0 & 1 \\
 1 & 1& 1 \\ \hline
 \end{array} \, ,
 \mbox{or} \,
 \begin{array}{  | c c  c|}
\hline
 1 & 1 & 0\\
 1 & 1& 1 \\ \hline
 \end{array}\, ,
 \end{equation}
depending on whether the variable is the first, second, or third variable in the clause with respect to our variable ordering.

\item For each intersection of a variable row and clause column, such that the variable is represented negatively in the clause, we have the measurement possibilities given by:
\begin{equation}
\begin{array}{  | c c  c|}
\hline
 1 & 1 & 1 \\
 0 & 1& 1 \\ \hline
 \end{array}\, ,
 \begin{array}{  | c c  c|}
\hline
 1 &1 & 1 \\
 1 & 0& 1 \\ \hline
 \end{array} \, ,
 \mbox{or} \,
 \begin{array}{  | c c  c|}
\hline
 1 & 1 & 1\\
 1 & 1& 0 \\ \hline
 \end{array}\, ,
 \end{equation}
depending on whether the variable is the first, second, or third variable in the clause with respect to our variable ordering.
\end{itemize}
We can see that choices of subrow in a measurement row is equivalent to an assignment of that variable, and that the clause structure effectively bans the assignment disallowed by that specific clause, for example the clause $(x_1 \vee x_2 \vee x_3)$ corresponds to the possibility table
\begin{equation*}
\begin{array}{c  | c c  c|}
\hline
 x_1=0&0 & 1 & 1 \\
 x_1=1&1 & 1& 1 \\ \hline
 x_2=0&1 &0 & 1 \\
x_2=1& 1 & 1& 1 \\ \hline
x_3=0& 1 & 1 & 0\\
x_3=1& 1 & 1& 1 \\ \hline
 \end{array}
 \end{equation*}
We see we cannot simultaneously assign a 1 to each of the top subrows as part of a deterministic grid. We, then, need only to answer the question of the necessary and sufficient nature of \textsc{2-Robust}ness; we note that choosing a specific 1 to complete to a grid represents therefore an assignment of one variable, and the structure of the table then might result in the forced selection of another variable's assignment, if there is a 0 in the same subcolumn as the selected 1. Therefore by the same logic as before, the \textsc{(2,3)-PossLoc} decision problem is \textbf{NP}-complete.
\end{proof}
Therefore, since \textsc{2-Robust 3-SAT} is \textbf{NP}-Hard; any instance of \textsc{2-Robust 3-SAT} can be reduced to an instance of \textsc{(2,3)-PossLoc} in polynomial time; and \textsc{(2,3)-PossLoc} is in \textbf{NP}, \textsc{(2,3)-PossLoc} is \textbf{NP}-Complete.
\end{proof}

\section{Quantum realisation}
We should note at this point that all the above proofs are for generalised no-signalling distributions. When we restrict ourselves to quantumly accessible distributions, all these problems become open.

\begin{thm}
All the quantumly accessible instances of the above \textsc{(2,3)-PossLoc} construction in which the party with three-outcome measurements has access only to a 2-dimensional quantum system are 0-valid and 1-valid under a variable renaming that can be calculated in linear time.
\end{thm}
\begin{proof}
Each measurement row corresponds to a two-outcome measurement. Without loss of generality, we can take these to be projective measurements since any two-outcome POVM is a convex mixture of such projectors. If our entangled state is pure, which it has to be in order for there to be any impossible events, then this means that each outcome steers Bob's system to one of two pure states that convexly mix to Bob's reduced state.

We note that each measurement column corresponds to a three-element POVM, and as such the hyperplane that is the convex hull of each triple contains the origin. Additionally, these POVM elements have an additional geometrical constraint: the POVM element with a 0 entry when $x_i=0$, say, must be orthogonal to the state to which Bob is steered when he gets the outcome $x_i=0$. This restricts each of these elements to be proportional to the projector onto the unique element orthogonal to the steered states.

We note that each variable is associated with two POVM elements, one representing positive occurrences of the variable, the other representing the negative occurrences of the variable; since the union of whose supports is the whole Hilbert space, at least one of these must have support in our chosen hypersphere and we choose one such POVM element to represent the positive occurrences. Now, if we chose our hemi-hypersphere such that no POVM elements lie on its equator, which is always possible since there are only finitely many POVM elements, then we see that each triple must have at least one POVM element from that half of the Hilbert space. Each clause therefore contains at least one positively-represented variable and so the all-true assignment of the variables is a satisfying assignment. Additionally, the opposite holds and so the all-false assignment of the variables is satisfied also after this transformation has taken place.
\end{proof}
We note that since our reduction in lemma \ref{nplem} was from a robust version of this problem, this proof alone does not demonstrate that the quantum realisation is not \textbf{NP}-complete; however it does demonstrate that we have access only to a very restricted set of problems we can embed, and therefore the hardness results of the previous section do not automatically carry through into the quantumly-accessible world. We shall see however, that when robustness is introduced into the mix, even simple problems can become \textbf{NP}-hard.

We see that we cannot make any assumptions about the difficulty of 2-robust variants of decision problems based on the difficulty of the underlying problem. We have also seen that we can embed some \textsc{2-Robust 0-valid 1-valid 3-SAT} problems into the quantum formalism. We note however that an arbitrary instance cannot be embedded: one can show that the array shown here is ruled out by the second tier of the NPA-hierarchy \cite{Nava2008}:

\begin{equation*} \label{badarray}
\begin{array}{  | c c  c| c c c|}
\hline
 0 & 1 & 1 & 0 & 1 & 1 \\
 1 & 1& 1& 1 & 1& 1 \\ \hline
 1 &0 & 1 & 1 &0 & 1 \\
 1 & 1& 1 & 1 & 1& 1  \\ \hline
 1 & 1 & 0&  1 & 1 & 1\\
 1 & 1& 1 &1 & 1 & 0 \\ \hline
 \end{array}
 \end{equation*}
 
Having considered some restrictions on which \textsc{3-SAT} instances can be embedded into \textsc{(2,3)-PossLoc}, we will now present a constructions will be given here that enables a reasonably large class of \textsc{3-SAT} instances to be embedded. We will consider the case in which our projective measurements act on a qubit, and are additionally confined to a single plane through the Bloch sphere. It is possible that relaxing either of these assumptions could allow more instances to be embedded, however this is not clear: moving from a two-dimensional to a higher-dimensional quantum system could allow more flexibility with regards to the geometry of the projectors, but since our measurements can only have two outcomes, having more dimensions causes the creation of impossible events to be more difficult; likewise moving from a system in which all the projectors are coplanar makes it harder for impossible events to be created from the projectors since for a qubit at least, we need three coplanar vectors to form a valid POVM. We note that while the array in equation \ref{badarray} is clearly not possible in this situation, since if projectors $\{P_1,P_2,P_3\}$ in a two-dimensional Hilbert space have the identity projector in their convex hull, then the set $\{P_1,P_2, P_3^\perp\}$ does not unless $P_1=P_2^\perp$, which would form a pair not a triple, this array is not possible to create in any quantum dimension and any realisation of measurement outcomes.

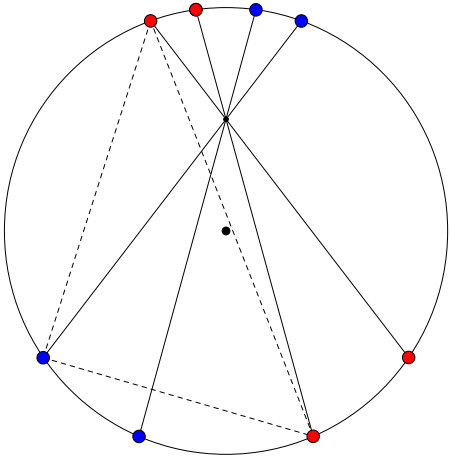
\begin{figure}[h]
\centering
\resizebox{0.45\textwidth}{0.45\textwidth}{
\begin{tikzpicture}[line cap=round,line join=round,>=triangle 45,x=1.25cm,y=1.25cm]
\clip(-5.,-5.) rectangle (5.,5.);
\draw(0.,0.) circle (5.cm);
\draw (-3.3,-2.27)-- (1.36,3.76);
\draw (0.54,3.96)-- (-1.57,-3.68);
\draw (1.5730879602280303,-3.6776887129535054)-- (-0.5438388074049572,3.962857473031342);
\draw (-1.3609931181462014,3.7613425438742856)-- (3.296294172570832,-2.2659313157895093);
\draw [dash pattern=on 3pt off 3pt] (-1.3609931181462014,3.7613425438742856)-- (-3.3,-2.27);
\draw [dash pattern=on 3pt off 3pt] (-3.3,-2.27)-- (1.5730879602280303,-3.6776887129535054);
\draw [dash pattern=on 3pt off 3pt] (1.5730879602280303,-3.6776887129535054)-- (-1.3609931181462014,3.7613425438742856);
\begin{scriptsize}
\draw [fill=black] (0.,0.) circle (2.5pt);
\draw [fill=black] (0.,2.) circle (1.5pt);
\draw [fill=ffqqqq] (3.296294172570832,-2.2659313157895093) circle (4pt);
\draw [fill=ffqqqq] (-1.3609931181462014,3.7613425438742856) circle (4pt);
\draw [fill=ffqqqq] (-0.5438388074049572,3.962857473031342) circle (4pt);
\draw [fill=ffqqqq] (1.5730879602280303,-3.6776887129535054) circle (4pt);
\draw [fill=ffqqqq] (-0.5438388074049572,3.962857473031342) circle (4pt);
\draw [fill=qqqqff] (-3.3,-2.27) circle (4pt);
\draw [fill=qqqqff] (-1.57,-3.68) circle (4pt);
\draw [fill=qqqqff] (0.54,3.96) circle (4pt);
\draw [fill=qqqqff] (1.36,3.76) circle (4pt);
\end{scriptsize}
\end{tikzpicture}}
\caption{A red positive literal, a red negative literal and a blue negative literal form a valid POVM.}
\label{Basicdiagram}
\end{figure}

We will draw inspiration for our projector geometry from the generalised Hardy paradox given above and shown in figure \ref{steer}. Any triple-outcome measurement for Bob will be enacted as a triple of projectors which include the identity operator in their convex hull. The two-outcome measurement is chosen such that the projector for the measurement subrow with the desired 0 steers the state into the one orthogonal to the projector assigned to relevant measurement subcolumn. Explicitly, our embedding has to conform to the following geometric restrictions:

\begin{itemize}
\item The reduced state $\rho_B$ is a point inside or on the edge of a circle.

\item A measurement row corresponds to a line going through the point corresponding to $\rho_B$. The two measurement outcomes are represented by the points at which this line intersects the circle, or, equivalently, the inversion of these points through the centre of the circle. The latter denotation will now be used.

\item A measurement column consists of three points around the edge of the circle such that the circle's centre is in their convex hull, each of which is associated with a single measurement subcolumn.

\item An impossible event happens when the point corresponding to that measurement subcolumn and the point corresponding to that measurment subrow are the same point.
\end{itemize}

A specific quantum scenario will now be explored alongside a characterisation of some of the problem instances that can be embedded into it. Given a \textsc{3-SAT} instance, let us assign each variable a colour: blue or red. As can be seen in figure \ref{Basicdiagram}, a clause consisting of a red positive literal, a red negative literal, and a blue negative literal,\emph{eg} $(\color{red}{l_1} \color{black}{\wedge} \color{red}{\neg l_2} \color{black}{\wedge} \color{blue}{\neg l_3})$, forms a valid POVM.

If we add projectors for $x$ and $y$ that are close to the eigendecomposition of $\rho_B$ (distinguished only because we want them to have independent possibilities), as shown in figure \ref{xydiag}, we can note that we can also support clauses of forms $(\color{red}{l_1} \color{black}{\vee} \color{blue}{ l_2} \color{black}{\vee} \neg x)$ or $(\color{red}{\neg l_1} \color{black}{\vee} \color{blue}{\neg l_2} \color{black}{\vee} y)$.

\begin{figure}[h]
\centering
\resizebox{0.45\textwidth}{0.45\textwidth}{
\begin{tikzpicture}[line cap=round,line join=round,>=triangle 45,x=1.0cm,y=1.0cm]
\clip(-5.,-5.) rectangle (5.,5.);
\draw(0.,0.) circle (4.cm);
\draw (-3.3,-2.27)-- (1.36,3.76);
\draw (0.54,3.96)-- (-1.57,-3.68);
\draw (1.5730879602280303,-3.6776887129535054)-- (-0.5438388074049572,3.962857473031342);
\draw (-1.3609931181462014,3.7613425438742856)-- (3.296294172570832,-2.2659313157895093);
\draw (0.,4.)-- (0.,-4.);
\draw (0.003084062551349804,4.35) node[anchor=north west] {$x$};
\draw (-0.4,4.35) node[anchor=north west] {$y$};
\begin{scriptsize}
\draw [fill=black] (0.,0.) circle (2.5pt);
\draw [fill=black] (0.,2.) circle (1.5pt);
\draw [fill=ffqqqq] (3.296294172570832,-2.2659313157895093) circle (1.5pt);
\draw [fill=ffqqqq] (-1.3609931181462014,3.7613425438742856) circle (1.5pt);
\draw [fill=ffqqqq] (-0.5438388074049572,3.962857473031342) circle (1.5pt);
\draw [fill=ffqqqq] (1.5730879602280303,-3.6776887129535054) circle (1.5pt);
\draw [fill=ffqqqq] (-0.5438388074049572,3.962857473031342) circle (1.5pt);
\draw [fill=qqqqff] (-3.3,-2.27) circle (1.5pt);
\draw [fill=qqqqff] (-1.57,-3.68) circle (1.5pt);
\draw [fill=qqqqff] (0.54,3.96) circle (1.5pt);
\draw [fill=qqqqff] (1.36,3.76) circle (1.5pt);
\draw [fill=uuuuuu] (0.,-4.) circle (1.5pt);
\draw [fill=uuuuuu] (0.,4.) circle (1.5pt);
\end{scriptsize}
\end{tikzpicture}}
\caption{Adding in $x$ and $y$ projectors close to the eigendecomposition of $\rho_B$.}
\label{xydiag}
\end{figure}
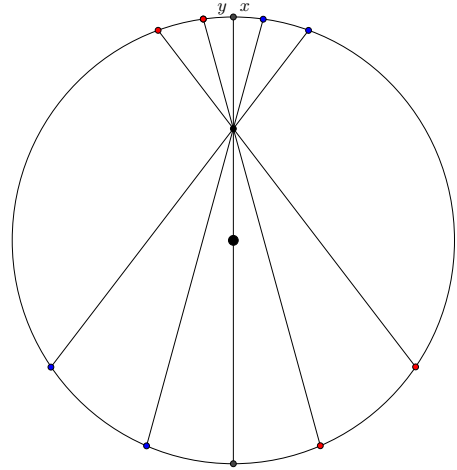

If we want to perform the transformation procedure in theorem \ref{rob3sat}, we can see from these constructions that the clauses in the initial \textsc{3-SAT} instance can have one of the following forms:

\begin{itemize}
\item $(\color{red}{\neg l_1} \color{black}{\vee} \color{blue}{\neg l_2} \color{black}{\vee} \color{red}{l_3} \color{black})$, since this becomes modified to 
$(\color{red}{\neg l_1} \color{black}{\vee} \color{blue}{\neg l_2} \color{black}{\vee} z)\wedge(\neg z \vee \color{blue}{ l_3}\color{black}{\vee} y)$, and by choosing the coloration $\color{red}{z}$ both the clauses have permitted forms.

\item $(\color{red}{ l_1} \color{black}{\vee} \color{blue}{l_2} \color{black}{\vee} \color{red}{\neg l_3}\color{black} )$, since this becomes modified to 
$(\color{red}{l_1} \color{black}{\vee} \color{blue}{ l_2} \color{black}{\vee} \neg z)\wedge(z \vee \color{blue}{ l_3}\color{black}{\vee} \neg x)$, and by choosing the coloration $\color{red}{z}$ both the clauses have permitted forms.

\item $(\color{red}{ l_1} \color{black}{\vee} \color{red}{l_2} \color{black}{\vee} \color{blue}{l_3}\color{black} )$, by choosing $z$ not with a colour but, like $x$ and $y$, close to the eigendecomposition of $\rho_B$.

\item $(\color{red}{\neg l_1} \color{black}{\vee} \color{red}{\neg l_2} \color{black}{\vee} \color{blue}{\neg l_3}\color{black} )$, by choosing $z$ not with a colour but, like $x$ and $y$, close to the eigendecomposition of $\rho_B$.
\end{itemize}
We also clearly by symmetry have as permissible clauses the images of the ones above under an exchange of the colours blue and red. The rules seem to be, then, that in the original \textsc{3-SAT} each variable must be colourable either red or blue such that each clause contains at least one literal of each colour, and that the literal in each clause that is the only one of its colour must not have opposite sense to the other two literals. The author has been unable to produce a proof of \textbf{NP}-hardness or membership in \textbf{P} of such \textsc{3-SAT} instances. In any case, the quantum realisation forces a very strong geometrical relationship on clauses of the embedded instance.
 
 \section{Conclusion}
 
 By demonstrating the \textbf{NP}-completeness of $(2,3)$-\textsc{PossLoc}, the computational complexity of all such possibilistic locality experiments have now been characterised. However, the problem of quantum realisation remains open; we have also seen that it is difficult even to rule out a quantum realisation under the restrictive assumption that the party with three-outcome measurements has access only to a two-dimensional quantum system-- although this does also imply without loss of generality that the entire state under questioning is an entangled state of two qubits as can be seen by invoking the Schmidt decomposition. A natural extension of this problem would be into the formalism of ontological models; it is possible to prove possibilistic nonlocality results at the level of underlying ontological models that nonetheless rely on operational probabilities rather than possibilities, as is the case in section 3 of reference  \cite{Jevt2015}.

\section*{Acknowledgements}

I would like to thank Terry Rudolph for multiple helpful discussions, and acknowledge support from Cambridge Quantum Computing Limited and from EPSRC \emph{via} the Centre for Doctoral Training in Controlled Quantum Dynamics..

\bibliography{/Users/andrewsimmons/Documents/Latex/Bib/bibliography}{}
\bibliographystyle{plain}

\end{document}